\documentclass[submission,copyright,creativecommons]{eptcs}
\providecommand{\event}{NCMA 2022} 
\usepackage{underscore}           
\usepackage{graphicx}
\usepackage{amsmath,amssymb,amsthm}
\usepackage{pifont}

\usepackage{enumitem}
\usepackage{array}
\usepackage{color}
\newcolumntype{P}[1]{>{\centering\arraybackslash}p{#1}}
\newtheorem{theorem}{Theorem}
\newtheorem{proposition}{Proposition}
\newtheorem{lemma}{Lemma}
\newtheorem{corollary}{Corollary}
\newtheorem{example}{Example}

\title{Quasi-deterministic $5'\rightarrow 3'$  Watson-Crick Automata}
\author{Benedek Nagy
\institute{Department of Mathematics, Faculty of Arts and Sciences, \\
      Eastern Mediterranean University, Famagusta, North Cyprus, via Mersin-10, Turkey\\}
      \email{nbenedek.inf@gmail.com}
}
\def\titlerunning{Quasi-deterministic $5'\rightarrow 3'$  Watson-Crick Automata}
\def\authorrunning{Benedek Nagy}

\begin{document}
\maketitle

\begin{abstract}
Watson-Crick (WK) finite automata are working on a Watson-Crick tape, that is, on a DNA molecule. A double stranded DNA molecule contains two strands, each having a $5'$ and a $3'$ end, and these two strands together form the molecule with the following properties.
The strands have the same length, their $5'$ to $3'$ directions are opposite, and in each position, the two strands have nucleotides that are complement of each other (by the Watson-Crick complementary relation).
 Consequently, WK automata have two reading heads, one for each strand. In traditional WK automata both heads read the whole input in the same physical direction, but in $5'\rightarrow 3'$  WK automata the heads start from the two extremes and read the input in opposite direction. In sensing  $5'\rightarrow 3'$  WK automata, the process on the input is finished when the heads meet,
and the model is capable to accept the class of linear context-free languages.
 Deterministic variants are weaker, the class named 2detLIN, a proper subclass of linear languages is accepted by them.
Recently, another specific variants, the state-deterministic sensing  $5'\rightarrow 3'$  WK automata are investigated in which the graph of the automaton has the special property that for each node of the graph, all out edges (if any) go to a sole node, i.e., for each state  there is (at most) one state that can be reached by a direct transition. It was shown that this concept is somewhat orthogonal to the usual concept of determinism in case of sensing  $5'\rightarrow 3'$  WK automata.
In this paper a new concept, the quasi-determinism is investigated, that is in each configuration of a computation (if it is not finished yet), the next state is uniquely determined although the next configuration may not be, in case various transitions are enabled at the same time. We show that this new concept is a common generalisation of the usual determinism and the state-determinism, i.e., the class of quasi-deterministic sensing  $5'\rightarrow 3'$  WK automata is a superclass of both of the mentioned other classes.
There are various usual restrictions on WK automata, e.g., stateless or 1-limited variants.
We also prove some hierarchy results among language classes accepted by various subclasses of
quasi-deterministic sensing  $5'\rightarrow 3'$  WK automata and also some other already known language classes.
\end{abstract}

\section{Introduction}
\label{s:intro}

DNA computing \cite{Adleman,Lipton,Paun} constitutes some of the most known natural computing paradigms. Some of these models based on the very small size of DNA coding and processing the information at molecular level. Some other paradigms heavily use the structure of the DNA molecules.
Watson-Crick automata  are belonging to the latter DNA computing models. In fact, they are finite state machines that are working on double stranded tape (like a DNA molecule). The symbols located in the same position  of the double-stranded tapes are related by the Watson-Crick complementarity relation.
Watson-Crick automata are abbreviated as WK automata for short based on the two extremes of the pair of names whom they are named.
 There are various models of WK-automata, in the classical models \cite{Czeizle,Freund,Paun,Sempere1,Sempere2}, the two heads belonging to the two strands are starting from the same end of the analysed DNA molecule and moving to the same physical direction. However,
the two strands of a DNA molecule have opposite $5'\rightarrow 3'$  orientation, i.e., their biochemical direction is opposite. The reverse and the $5'\rightarrow 3'$  variants of WK automata  are more realistic in the sense, that both heads use the same biochemical (that is opposite physical) direction  \cite{Freund,NaCo,Leupold,DNA13,DNA2008}. Some variants of the reverse Watson-Crick automaton with sensing parameter which tells whether the upper and the lower heads are within a fixed distance (or meet at the same position) are discussed in \cite{DNA2008,Nagy2009,iConcept,Nagy2013}. As usual, automata are closely connected to formal language theory. Later, it was shown that the sensing  $5'\rightarrow 3'$  WK automata with sensing parameter are equivalent to a newer model without the sensing parameter \cite{NaCo-Shag,AFL}, moreover both models characterise the class LIN of linear context-free languages.

There are three concepts of determinism introduced to WK-automata \cite{detWK1,detWK2,detWK3},  due to the possible properties of the used Watson-Crick complementary relations, each of those concepts differs from the other two.
 According to the definition, in a \textit{weakly deterministic} WK-automaton in any configuration that could occur in any computation, there is at most one way to continue the computation. We believe that this concept is the closest to the original idea of determinism.
 However, in traditional WK automata, other related concepts may also be used. The second concept, the \textit{deterministic} WK-automata, have a stronger constraint than the former weakly deterministic ones, namely: (the formal description will be explained in the next section) if there are two transitions  from the same state, i.e., the automaton has both $p_1\in \delta(q,u_1,v_1)$ and  $p_2\in \delta(q,u_2,v_2)$ with some states $q,p_1,p_2$ and strings $u_1,u_2,v_1,v_2$, then none of $u_1$ and $u_2$ are prefix of each other and none of $v_1$ and $v_2$ are prefix of each other.
Finally, the third concept is defined as follows. A WK-automaton is \textit{strongly deterministic} if it is deterministic (as above), moreover the identity is used as Watson-Crick complementarity relation.

In sensing  $5'\rightarrow 3'$  WK automata the process on the input is finished when the heads meet, thus each position of the Watson-Crick tape is read by only one of the heads. Due to this fact, the complementarity relation does not really play importance in these models.  On the other hand, the complementarity can also be excluded from the classical models \cite{Kuske}, however, as we have mentioned it may play an important role in defining various types of determinism and also in complexity point of view.
 Deterministic variants of  $5'\rightarrow 3'$  WK automata are described in detail in \cite{ActInf2020,UCNC18}, they accept the sublinear class of languages 2detLIN. (As, the  sensing  $5'\rightarrow 3'$  WK automata are 2-head automata models, and they accept the linear languages,  the class 2detLIN is named in this way as it is the deterministic counterpart of the class of linear languages by this accepting model. It should be noted that 2detLIN is incomparable with detLIN, the class of linear languages accepted by deterministic one-turn pushdown automata, under set-theoretic inclusion.)
WK automata have restricted variants based on restrictions on the states and/or on the transitions. These restrictions are orthogonal to the concept of determinism.
Recently, a new and related concept, the \textit{state-determinism}, was also introduced (for finite and sensing $5'\to 3'$ WK automata, see \cite{stateDET-NaCo}). Here, we present another new concept, the \textit{quasi-determinism}. In these new models, the state of the next configuration is determined, but the next configuration may not be.
In the state-deterministic automata, the state of the next configuration depends only on the actual state, while in the quasi-deterministic automata (as we will define it formally), it may also depend on the part of the input being read.

In this paper, we consider a new type of concept of determinism, namely the quasi-determinism, which is closely related to the usual determinism in the case of finite automata. Actually, as we will see, $\lambda$-transition free quasi-deterministic NFA = DFA. On the other hand, already for NFA with $\lambda$-transitions we may allow some non-determinism that is similar to the earlier introduced state-determinism.

We also show that quasi-deterministic sensing $5'\rightarrow 3'$  WK automata are strictly more powerful than deterministic sensing $5'\rightarrow 3'$  WK automata and
state-deterministic sensing $5'\rightarrow 3'$  WK automata. However, they are not as powerful as nondeterministic sensing $5'\rightarrow 3'$  WK automata.
  Some hierarchy results among the language classes defined by the known restricted
variations such as
all-final, simple, 1-limited, and stateless $5'\rightarrow 3'$  Watson-Crick automata are shown together with some relations to REG of regular, LIN of linear context-free languages, and 2detLIN of the class accepted by deterministic sensing $5'\rightarrow 3'$  WK automata.

%
\section{Basic definitions}
\label{s:pre}
We assume that the reader is familiar with basic concepts of formal languages and automata theory,
the language classes of the Chomsky hierarchy and regular expressions,
otherwise, she or he is referred to \cite{HopUl,Handb}.
In this paper, we denote the alphabet by $T$ and the empty word by $\lambda$.

We briefly  recall the concepts of finite automata and sensing $5'\to 3'$ Watson-Crick automata.

A finite automaton is  a 5-tuple $A=(T,Q,q_0,F,\delta)$, where:
\begin{itemize}
\item $T$ is the (input) alphabet,
\item $Q$ is the finite set of states,
\item $q_0\in Q$ is the initial state,
\item $F\subseteq Q$ is the set of final (also called accepting) states and
\item $\delta$ is the transition mapping.
    \end{itemize}
If $\delta$ is written in the form  $\delta: Q \times \left(T \cup{\lambda} \right)\rightarrow 2^Q$, then $A$ is a nondeterministic finite automaton with (allowed) transitions by the empty word (NFA+$\lambda$, for short).  If $\delta$ is also in the form  $\delta: Q \times T \rightarrow 2^Q$, then $A$ is nondeterministic finite automaton without transitions by the empty word (NFA for short). Further, if $\delta: Q \times T \rightarrow Q$ is a possibly partially defined function, then $A$ is a deterministic finite automaton (DFA).

The initial configuration of a finite automaton $A$ consists of the initial state and the input word $w\in T^*$ as a pair $(q_0,w)$. The computation on an input word goes by configurations according to the transition mapping as follows: $(q,au) \Rightarrow (q',u)$ if  $q'\in\delta(q,a)$. A word $w$ is accepted if $(q_0,w) \Rightarrow^* (q_F,\lambda)$ for a state $q_F\in F$, where $\Rightarrow^*$ is a reflexive and transitive closure of $\Rightarrow$. The set of accepted words form the accepted language $L(A)$.

Notice that, as usual, the finite automata can process the input from left to right and (at most) one letter is being read in each transition (depending on the type, see above).
However, in the literature there are various cases, when  string reading is allowed, i.e., from a state a finite set of words are given for which transitions are allowed. It is well-known that this feature does not increase the ``accepting power'' of finite automata, still exactly the class of regular languages can be accepted, however, it may have some effects on complexity measures, etc.
Also there are various finite state automata models, when the automaton may read and process the input not in the strict left to right manner, e.g., automata with translucent letters \cite{CR,tr-let,JCSS,newFF} and jumping automata \cite{JUMP-Sz,
jump}.

Now we are turning to give our central concept, the sensing $5'\rightarrow3'$ WK automata a model that is also capable to process the input not only in left to right manner.

The two strands of the DNA molecule have opposite $5'\rightarrow3'$ orientations, consequently, it is natural to consider Watson-Crick finite automata that parse the two strands of the Watson-Crick tape in opposite directions. Since in sensing $5'\rightarrow3'$ WK automaton the heads sense that they are meeting and the process on the input is finished (at the latest) in this position of the heads, the complementarity relation does not really play any role %
 in this model, w.l.o.g., the identity can be used. Therefore we present a simplified (but equivalent) model in which a normal tape is used for the input. Moreover, since the head starting from the left extreme is always to the left of the other head, but the meeting final position, we may refer to the heads as left and right heads, or alternatively, as first and second heads.

Formally, a Watson-Crick automaton is considered to be a 5-tuple $A=(T,Q,q_0,F,\delta)$ similarly to the finite automata, with
\begin{itemize}
\item the (input) alphabet $T$, originally the letters standing for possible bases of the nucleotides,
\item the finite set of states $Q$,
 the initial state $q_0\in Q$ and the set of final states
$F\subseteq Q$,
\item the transition mapping $\delta$  is of the form  $\delta: Q \times T^{*} \times T^{*} \rightarrow 2^Q$, such that it is non-empty only for finitely many triplets $(q,u,v), q \in Q, u,v\in T^*$.
    \end{itemize}
Notice that there are two main differences between finite automata and Watson-Crick automata and both of them are in their transition mappings: the Watson-Crick automata have two reading heads and, further, a Watson-Crick
automaton may read strings in a transition.

\begin{figure}[h]
    \centering
        \includegraphics[scale=0.3]{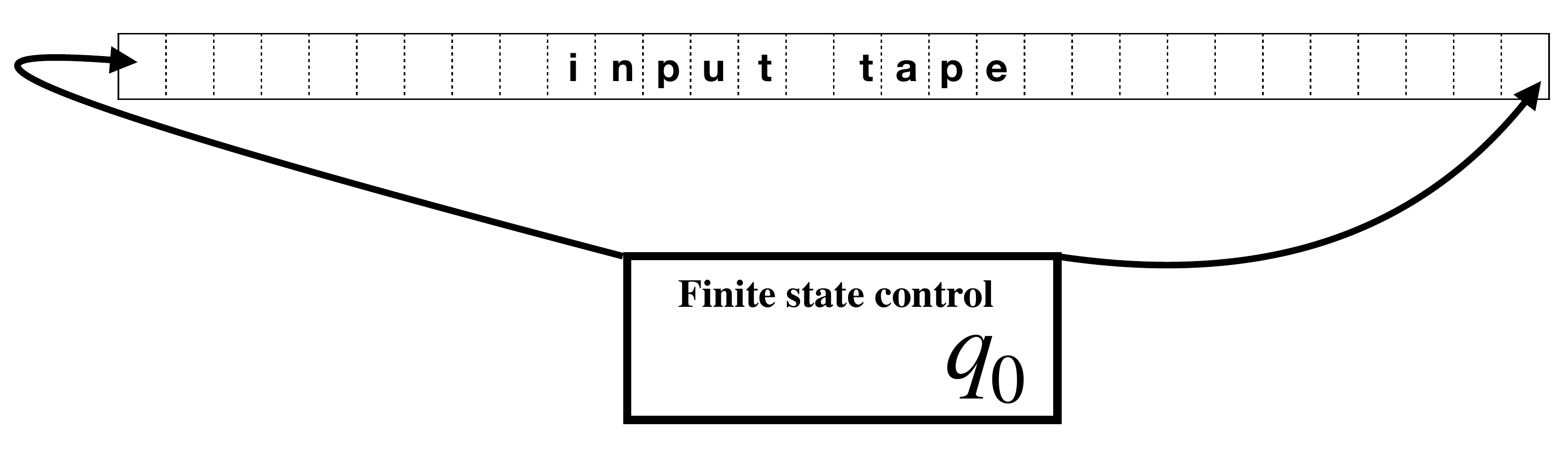}         \includegraphics[scale=0.3]{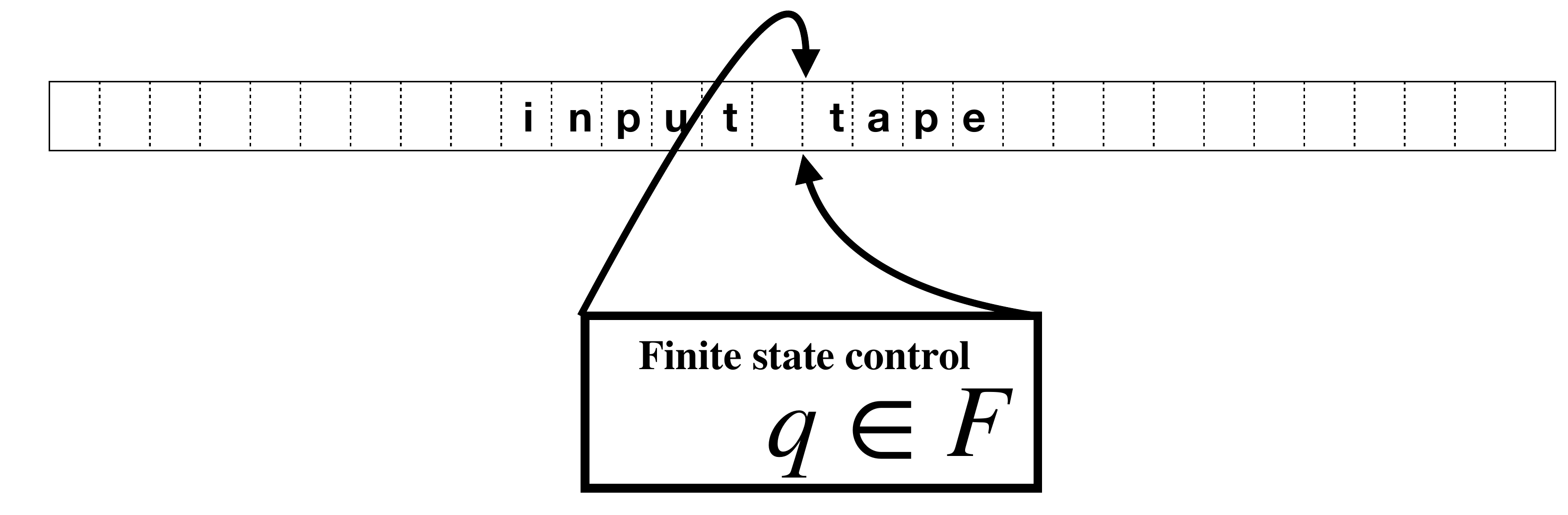}
            \caption{Schematic pictures of a sensing $5'\rightarrow 3'$ WK automaton in the initial configuration (top) and in an accepting configuration (bottom, with a final state $q$).}
    \label{sd}
\end{figure}

A configuration of a Watson-Crick automaton is a pair $(q,w)$ where $q$ is the current state of the automaton and $w$ is the part of the input word which has not been processed (read) yet. In sensing $5' \rightarrow 3'$ WK automata,
for any $w',x,y\in T^*$, $q,q'\in Q$, we write a transition (step of the computation) between two configurations as follows:
$(q,xw'y)\Rightarrow(q',w' )$ if and only if $q'\in \delta(q,x,y)$. We denote the reflexive and transitive closure of the relation $\Rightarrow$ by $\Rightarrow^*$. Therefore, for a given input $w\in T^*$, an accepting computation is a sequence of transitions $(q_0,w) \Rightarrow^* (q_F,\lambda)$, starting from the initial state and ending in a final state.

The language accepted by a WK automaton $A$ is:
$L(A)=\{ w\in T^*  \mid  (q_0,w) \Rightarrow^* (q_F,\lambda),  q_F\in F \}.$

Figure \ref{sd} shows schematic representations of an initial configuration and an accepting configuration of a sensing $5'\rightarrow3'$ WK automaton.

Since in DNA computing the empty word does not represent any molecule, in this paper, as usual in this field, we do not care about if a language contains the empty word or not, that is, we call two automata equivalent if they accept the same language modulo the empty word.
Further, the class of languages accepted by
sensing $5'\rightarrow 3'$ WK automata is exactly the class of linear context-free languages of the Chomsky hierarchy (\cite{DNA2008,Nagy2009,Nagy2013,NaCo-Shag,AFL}).

There are  restricted variants of WK automata which are defined as follows.
We say that a WK automaton is %
\begin{itemize}
\item
 \textit{stateless}, i.e., with only one state, if $Q=F=\{q_0\}$. This class of WK automata is denoted by the letter {N}, as \textbf{N}o-state.
\item %
 \textit{all-final}, i.e., with only final states, if $Q=F$.
This class of WK automata is denoted by the letter {F}, as all-\textbf{F}inal.
\item %
 \textit{simple}, when at most one head moves in a step, formally: $\delta:(Q\times ((\lambda,T^* )\cup(T^*,\lambda)))\rightarrow 2^Q$.
This class of WK automata is denoted by the letter {S}, as \textbf{S}imple.
\item %
 1-\textit{limited}, when exactly one letter is being read in each step of the computations, i.e., $\delta:(Q\times ((\lambda,T)\cup (T,\lambda)))\rightarrow2^Q$.
This class of WK automata is denoted by {1}, as \textbf{1}-limited.
\end{itemize}
Obviously, by definition, every N sensing $5'\rightarrow 3'$ WK automaton is an F sensing $5'\rightarrow 3'$ WK automaton, and also, every 1 sensing $5'\rightarrow 3'$ WK automaton is an S sensing $5'\rightarrow 3'$ WK automaton. %
There are also independent restrictions;  their combinations are
the {F1}, {N1}, {FS}, {NS} WK automata.

There is one usual definition of (classical) determinism in the case of sensing $5'\to 3'$ WK automata:
We say that a $5'\to 3'$ WK automaton is \textit{deterministic}, if there is at most one transition (computation step) is applicable in each configuration.
 It is easy to see that this condition implies the fact that
   at each possible configuration $(q,w)$ in any computation ($q\in Q, w\in T^*$), there is at most one configuration $(p,u)$ such that $(q,w) \Rightarrow (p,u)$.
Somewhat similarly, as at DFA, when deterministic WK automata are mentioned, we may use the letter D to denote this fact. Remember that the language class accepted by deterministic
sensing $5'\rightarrow 3'$ WK automata is a proper subset of LIN and denoted by 2detLIN \cite{ActInf2020,UCNC18}.

Now, we recall a related concept, the concept of state-determinism from \cite{stateDET-NaCo}.
An automaton (finite or Watson-Crick) is \textit{state-deterministic} if for each of its state $q\in Q$,
if there is a transition from $q$ and it goes to state $p$ (i.e., $p\in\delta(q,a)$ in the case of finite automata with $a\in T\cup\{\lambda\}$ and  $p\in\delta(q,u,v)$ in case of WK-automata with $u,v\in T^*$), then every transition from $q$ goes to $p$.
We may say that if an automaton (finite or Watson-Crick) has state $q$ in its configuration, then, if the process continues, the state of the next configuration is uniquely determined and it is $p$. State-deterministic automata will be denoted by the prefix sD.

Now, we are ready to define our new concept that plays a central role in this paper.
An automaton (finite or Watson-Crick) is \textit{quasi-deterministic} if the following condition holds. For each its possible configuration $(q,w)$ (with $q\in Q, w\in T^*$),
if $(q,w) \Rightarrow (p,u)$ and $(q,w) \Rightarrow (r,v)$ then $p=r$ must hold.
 In other words, there is a unique state $p$ such that
if there is a transition allowed in the configuration $(q,w)$, then it must go to a configuration in which the state is $p$.
For quasi-deterministic automata, we use the qD prefix in the sequel.

Automata are usually represented by their graphs, thus we may freely use graph theoretical concepts here.
Further, we assume that each state of the automaton appears in an accepting path of a word of the accepted language. A transition from a state to itself is represented by a loop edge.
Clearly, in a stateless automaton, there are only loop edges (if any).

\section{On quasi-deterministic finite automata}

First, for analogy, we show how this new concept can be understood in the case of traditional finite state automata. %

The next proposition shows how quasi-determinism relates to determinism when finite automata are considered.

\begin{proposition}
Let $A$ be a $\lambda$-transition free NFA. If $A$ is quasi-deterministic, then it is deterministic. Moreover, each DFA is a $\lambda$-free NFA that is quasi-deterministic.
\end{proposition}
\begin{proof}
  Since $\lambda$-transition free NFA read exactly one input letter in each transition,
     for a quasi-determin\-istic NFA,
   the first letter of the remaining input in the configuration and the actual state together uniquely determines the next state exactly in the same way as in a DFA by definition.
\end{proof}

The new concept also relates to state-determinism in finite automata:

\begin{proposition}
Let $A$ be a state-deterministic (NFA+$\lambda$) finite automaton, then $A$ is also quasi-determin\-istic.
\end{proposition}
\begin{proof}
  Obvious from the definitions of these models.
\end{proof}

As we have seen, quasi-determinism is a close connection to determinism and also to state-deter\-minism. Now to show that they are not exactly the same, %
we present 
an example.

\begin{example}\label{ex1}
Consider the finite automaton shown in Figure~\ref{fig-ex1}.
\begin{figure}[h]
    \centering
        \includegraphics[scale=0.75]{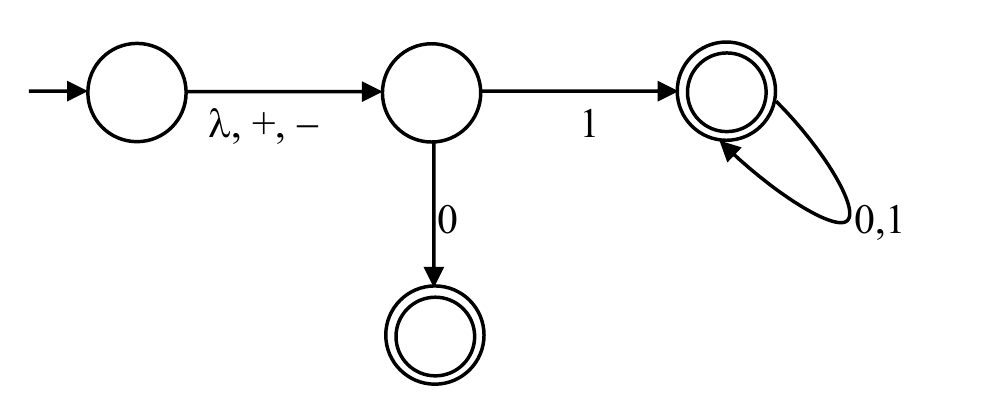}
            \caption{A finite automaton that is not deterministic, but quasi-deterministic.}
    \label{fig-ex1}
\end{figure}
It accepts the language of binary integers over the alphabet $\{+,-,0,1\}$.
On the one hand, it is not a DFA due to the transition by $\lambda$ from its initial state.
Moreover, it is not state-deterministic, as there is a state for which reading $0$ and $1$ lead to different states.
On the other hand, it is easy to check that it is quasi-deterministic:
 from the initial state, all the three possible transitions (including the one with $\lambda$) go to the same state. From the other states, the transitions are deterministic transitions.
\end{example}

Now, let us give a technical result about quasi-deterministic
finite automata. In fact, their states behave similarly to the states of the deterministic or to the states of state-deterministic %
automata.

\begin{lemma}
Let $A = (T,Q,q_0,F,\delta)$ be a quasi-deterministic
 finite %
NFA+$\lambda$ automaton. Then, the set of states can be partitioned into two sets as
$Q=Q_d \cup Q_s$ ($Q_d \cap Q_s = \emptyset$)
 such that
$$Q_d = \{q~|~ \emptyset =\delta(q,\lambda) \text{ and for any } a\in T, \text{ there is at most one state } p \text{ such that } p\in \delta(q,a) \}$$ and
$$Q_s = \{q~|~ \{p\} = \delta(q,\lambda) \text{ and } %
\text{ if } q' \in \delta(q,a) \text{ for some } a\in T, \text{ then } q' = p\}.$$
\end{lemma}
\begin{proof}
Let $A$ be an NFA+$\lambda$ automaton. Then, for any of its states $q\in Q$ there are two options, either $\emptyset \ne \delta(q,\lambda)$ or $\emptyset =\delta(q,\lambda)$.
In the latter case, there could be transitions from $q$ only by reading an input letter. If $A$ is quasi-deterministic, then from any configuration (with state $q$ having the first letter $a$ of the remaining input), the computation should go to a determined state (to $p$ in this case), if any (otherwise, the computation gets stuck and not accepting).
In the former case, by not reading any input letter, $A$ may enter a configuration $(p,w)$ from $(q,w)$ if $p\in\delta(q,\lambda)$. Since $A$ is quasi-deterministic, it may not happen that there is a $p'\in Q$ such that $p'\in\delta(q,\lambda)$ and $p\ne p'$.
Further, if there are transitions from $q$ by reading a letter, let us say $a\in T$, it must also go to state $p$.
\end{proof}

If $Q=Q_d$ in the previous lemma, then, in fact, the quasi-deterministic $A$ is also deterministic. Moreover, if $A$ is deterministic, then it is a quasi-deterministic finite automaton with the property $Q=Q_d$.
On the other hand, if $A$ is quasi-deterministic with $Q=Q_s$,
then $A$ is also state-deterministic. However, for a state-deterministic $A$, it may happen that it is quasi-deterministic, but $Q_d \ne \emptyset$.
 In this case, there are some states $q$ such that $\lambda$-transitions are not defined on those, but all the transitions from $q$ go to a unique state $p$ of the automaton.

As both the classes of DFA and NFA+$\lambda$ accept exactly the class of regular languages,
it is straightforward to see that the class of quasi-deterministic finite automata, denoted by qDFA, is also characterizing the same class of languages. From this point of view, the new model does not seem to give a new power.  However, from a complexity point of view, the classes of DFA and NFA have a remarkable difference, as in some cases, exponentially more states are needed for a DFA to accept the same regular language as an NFA needs. As qDFA allows a little non-determinism, it is an interesting challenge to see how the descriptional complexities of DFA, qDFA, NFA and NFA+$\lambda$ are related to each other. We leave this topic for a future work and in this paper, we concentrate on $5'\to 3'$ WK automata, in which the quasi-determinism
gives more freedom due to the two heads.

Observe that if a sensing $5'\rightarrow 3'$ WK automaton has transitions only with its left head (the right head is always reading $\lambda$), then, in fact, it is equivalent to a finite automaton. If this finite automaton is quasi-deterministic, then so is the $5'\rightarrow 3'$ WK automaton. We may use this link between finite automata and WK automata in the next section.

\section{On quasi-deterministic sensing $5'\to 3'$ WK automata}
\label{s:WK}

We start this section by showing that quasi-determinism is a generalisation of both
determinism and state-determinism also in the case of sensing $5'\rightarrow 3'$ WK automata.
The following statements are %
 consequences of the definitions of the used types of determinism.

\begin{proposition}
Let $A$ be a deterministic sensing $5'\rightarrow 3'$ WK automaton, then $A$ is quasi-deterministic.
\end{proposition}
\begin{proof}
In deterministic sensing $5'\rightarrow 3'$ WK automata, for each possible configuration of a computation, there is at most one next configuration, and, thus, its state is the only state that can be reached from the given configuration in one computation step.
\end{proof}

\begin{proposition}
Let $A$ be a state-deterministic sensing $5'\rightarrow 3'$ WK automaton, then $A$ is quasi-determin\-istic.
\end{proposition}
\begin{proof}
In state-deterministic sensing $5'\rightarrow 3'$ WK automata, for each possible configuration of a computation, there is at most one state that can appear in any of the next configurations, and, thus, this unique state (if any) is the one that proves the quasi-determinism.
\end{proof}

At this point, we can deduct that the language class accepted by  quasi-deterministic
sensing $5'\rightarrow 3'$ WK automata contains 2detLIN, the language class accepted by
deterministic sensing $5'\rightarrow 3'$ WK automata and also the language class
accepted by state-deterministic sensing $5'\rightarrow 3'$ WK automata. Now we show that both of these inclusions are proper.

\begin{example}\label{ex-stD-nonD}
  Let us consider the language $L_o=\{a^m  b^n ~|~ m\leq n\leq 2m\}$. This language is accepted by a $5'\rightarrow 3'$ WK automaton having only one state and two loop transitions with $(a,b)$ and $(a,bb)$. Obviously this automaton is quasi-deterministic, always only its sole state can appear in any next configurations. On the other hand, this language is not in 2detLIN, a deterministic sensing $5'\rightarrow 3'$ WK automaton has no chance to guess how the numbers of $a$-s and $b$-s are related to each other. More formally, the proof may go by contradiction. Let us assume that a deterministic
sensing $5'\rightarrow 3'$ WK automaton $A$ with initial state $q$ accepts $L_o$. Let $k$ be the number of states of $A$ and let $r$ be the maximal length of the words $A$ may read in a transition (radius of $A$). Let $m=3kr$. Consider the words $a^m b^m$ and $a^mb^{2m}$. Both are in $L_o$, however their accepting computation must start exactly in the same way in the first $2k$ steps ($A$ may read at most $2kr$ $a$-s by the first head and at most that many $b$-s by the second head during this part of the computation). Thus there is a state
which appears more than one configurations in this part of the computation, let it be $p$ (may be the same as $q$). This part of the computations on the two above words can be written as $(q,a^m b^m) \Rightarrow^* (p,a^{m-i_1}b^{m-j_1}) \Rightarrow^* (p,a^{m-i_2}b^{m-j_2}) \Rightarrow^* (f_1,\lambda)$
and $(q,a^m b^{2m}) \Rightarrow^* (p,a^{m-i_1}b^{2m-j_1}) \Rightarrow^* (p,a^{m-i_2}b^{2m-j_2}) \Rightarrow^* (f_2,\lambda)$ with accepting states $f_1,f_2$.
 Let us analyse the relation of $i_2-i_1$ and $j_2-j_1$, i.e., the number of $a$-s and $b$-s read in the cycle.
\begin{itemize}
\item If $i_2-i_1 > j_2-j_1$, i.e., more $a$-s are read than $b$-s, then we also have the accepting computation $(q,a^{m+i_2-i_1} b^{m+j_2-j_1}) \Rightarrow^* (p,a^{m+i_2-2i_1}b^{m+j_2-2j_1}) \Rightarrow^* (p,a^{m-i_1}b^{m-j_1}) \Rightarrow^* (p,a^{m-i_2}b^{m-j_2}) \Rightarrow^* (f_1,\lambda)$ contradicting to the fact that the word $a^{m+i_2-i_1} b^{m+j_2-j_1}$ is not in $L_o$.
\item In the case $i_2-i_1 = j_2-j_1$, i.e., the computation reads the same number of $a$-s and $b$-s in the cycle, we have the accepting computation
    $(q,a^{m-(i_2-i_1)} b^{2m-(i_2-i_1)}) %
    \Rightarrow^* (p,a^{m-i_2}b^{2m-j_2}) \Rightarrow^* (f_2,\lambda)$.
    However, in this case, $m-(i_2-i_1) = m-i_2+i_1$ is less than the half of $2m-(i_2-i_1) = 2m -i_2+i_1$, and thus, $a^{m-(i_2-i_1)} b^{2m-(i_2-i_1)}$ is not in $L_o$.
\item Finally, if $i_2-i_1 < j_2-j_1$, then there is an accepting computation
   $(q,a^{m-(i_2-i_1) }b^{m-(j_2-j_1) }) \Rightarrow^* %
   (p,a^{m-i_2}b^{m-j_2}) \Rightarrow^* (f_1,\lambda)$, however, the word $a^{m-(i_2-i_1) }b^{m-(j_2-j_1) }$ contains more $a$-s than $b$-s showing the contradiction in this case.
\end{itemize}
\end{example}

\begin{example}\label{ex-bab+b} Let us consider the regular
 language  $ b^*ab^* + b^* $. We show that it is not accepted by any state-deterministic
sensing $5'\rightarrow 3'$ WK automata. The proof goes by contradiction, thus,
assume that there is a state-deterministic sensing $5'\rightarrow 3'$ WK automaton $A$ accepting
$ b^*ab^* + b^* $. Since the language is infinite, $A$ must have a cycle. Definitely, $A$ has a transition where letter $a$ (or a string containing letter $a$) is being read by either head and let us refer to this head by $h$. Considering this transition, it cannot be in the cycle, since it would allow to read more than one $a$ in an accepted word. However, in a state-deterministic automaton,
 there are only finitely many states that are not in the cycle, and these states can only be visited during a computation only before the computation enters to the cycle. Thus, the mentioned transition must be from a state which can be visited before the computation enters into the cycle, let us refer to this state by  $p$. However, in the finitely many possible computations  connecting the initial state $q_0$ to $p$, head $h$ can read at most a given finite number, let us say $j$, of $b$-s (depending both on the numbers of states in this path, and also on the length of the read words of the transitions of the path). But, then it would be impossible for $A$ to accept a word such that its prefix (if $h$ is the left head) or its suffix (in case $h$ is the right head) has more $b$-s (before/after the $a$, respectively) than  $j$. However, as the language $ b^*ab^* + b^* $ has such words,  a contradiction is obtained, and the proof is finished.

Now, on the other hand, the language $ b^*ab^* + b^* $ is a regular language, and thus, it is  accepted by some DFA. And a DFA can also be seen as a specific quasi-deterministic
sensing $5'\rightarrow 3'$ WK automaton obtaining the second part of the proof.
\end{example}

Based on the previous propositions and examples, we have proven the following result.

\begin{theorem}
  The class of languages accepted by  quasi-deterministic
sensing $5'\rightarrow 3'$ WK automata is a proper superset of both 2detLIN and the class of languages accepted by  state-deterministic
sensing $5'\rightarrow 3'$ WK automata.
\end{theorem}

As we already mentioned, in the case of finite automata, quasi-determinism may play only a role in complexity, as, in fact, the class REG of regular languages is characterized by
DFA, NFA, NFA+$\lambda$ and also by qDFA. Contrary to this, we have seen that qD
sensing $5'\rightarrow 3'$ WK automata are more powerful than deterministic sensing $5'\rightarrow 3'$ WK automata. Now, we show that qD sensing $5'\rightarrow 3'$ WK automata
are still not as powerful as the general nondeterministic sensing $5'\rightarrow 3'$ WK automata.

\begin{theorem}
  The class of languages accepted by  quasi-deterministic
sensing $5'\rightarrow 3'$ WK automata is a proper subset of LIN, the class of languages accepted by  nondeterministic
sensing $5'\rightarrow 3'$ WK automata.
\end{theorem}
\begin{proof}
  The inclusion is trivial from the definition; we need to show only properness.
  Let us consider the language $L=\{a^nb^n ~|~ n\geq 0\} \cup \{a^nb^{2n} ~|~ n\geq 0\}$. Clearly, it is a union of two linear context-free languages, and it is also linear.
  The rest of the proof goes by contradiction. Thus, let us assume that there is a
quasi-deterministic
sensing $5'\rightarrow 3'$ WK automaton $A$ that accepts $L$. (W.l.o.g., we assume that there are no states in $A$ such that the only transition is the one that reads no letters from the input.)
Let the number of states of $A$ be $k$ and let the longest string read by a possible transition have length $\ell$. Now let us consider the words $u=a^{(3k+2)\ell}b^{(3k+2)\ell}$ and $v=a^{(3k+2)\ell}b^{(6k+4)\ell}$, clearly both of them are in $L$, thus $A$ must accept both of them.
On the other hand, $A$ may read only the prefixes and suffixes of those words, and thus cannot distinguish them. We shall show that by accepting both of these words by $A$ we can also accept some words that are not in $L$. Since the words are long enough, there must be a state during the computation that is repeated, more precisely, we have the following. As one step of the computation may process at most $2\ell$ letters (in fact, at most $\ell$ by the first head and at most $\ell$ by the second head), the accepting computation on $u$ (and also on $v$) must take more than $3k$ steps. By the pigeon-hole principle, there must be a state that appears at least $3$ times during such accepting computation (already in the first $2k+1$ configurations, i.e., during or right after the first $2k$ steps).
Now, in the first $3k$ steps, both in $u$ and $v$, the first head can only read $a$-s, i.e., in each step, a word from $a^*$ is read with length at most $\ell$. Similarly, the second head can read only word of $b^*$ with length at most $\ell$. Since $A$ is quasi-periodic, in each of these first $3k$ steps, the computation goes through the same sequence of states.
 In this way, if $q$ is a state that is used in at least three different configurations in the first $3k$ steps (we already have shown that such a state exists, maybe $q=q_0$), then we
 have an accepting computation on $u$ as $(q_0,u) \Rightarrow^* (q,a^{i_1}b^{j_1})
\Rightarrow^* (q,a^{i_2}b^{j_2}) %
\Rightarrow^*
 (\lambda,q_f)$ with some accepting state $q_f\in F$. By the cycle of the computation made from state $q$ to reach again state $q$, it is clear that the word $u' = a^{(3k+2)\ell+n(i_2-i_1)}b^{(3k+2)\ell+n(j_2-j_1)}$  is also accepted for all positive integer $n$, which implies that $i_2 - i_1 = j_2 - j_1$. Let this number be denoted by $x$, i.e., $x=i_2-i_1$. In such a cycle, $A$ reads $x$ $a$-s and $x$ $b$-s, and $x>0$.
Now, considering the accepting computation on $v$ which also goes through on $q$ more than once during the first $3k$ steps, i.e., it can be written as $(q_0,v) \Rightarrow^* (q,a^{i_3}b^{j_3})
\Rightarrow^* (q,a^{i_4}b^{j_4}) \Rightarrow^* %
 (\lambda,q'_f)$ with an accepting state $q'_f$.
Here, there is also a cycle in the computation determined by the above two configurations containing state $q$. However, if one insert the previously studied cycle of the accepting computation on $u$, here, then we get a computation that is also accepting on the word $w=a^{(3k+2)\ell+x}b^{(6k+4)\ell+x}$
  $(q_0,w) \Rightarrow^* (q,a^{i_3+x}b^{j_3+x}) \Rightarrow^* (q,a^{i_4+x}b^{j_4+x})
\Rightarrow^* (q,a^{i_4}b^{j_4}) \Rightarrow^* %
 (\lambda,q'_f)$. However, the numbers of $a$-s and $b$-s in $w$ are neither equal (as $\ell>0$),
 nor the number of $b$-s is double than the number of $a$-s (as $x>0$).
 This contradicts to the fact that $A$ accepts $L$.
\end{proof}

From the proof of the previous result, we may also infer the following.

\begin{corollary}
  The class of languages accepted by quasi-deterministic sensing $5'\rightarrow 3'$ WK automata
  is not closed under union.
\end{corollary}

Based on the results presented so far,
Figure \ref{hasse} shows how the language class accepted by quasi-deterministic
sensing $5' \rightarrow 3'$ WK automata relates to the language classes accepted by the previously studied related models.
\begin{figure}[t]
    \centering
        \includegraphics[scale=.83]{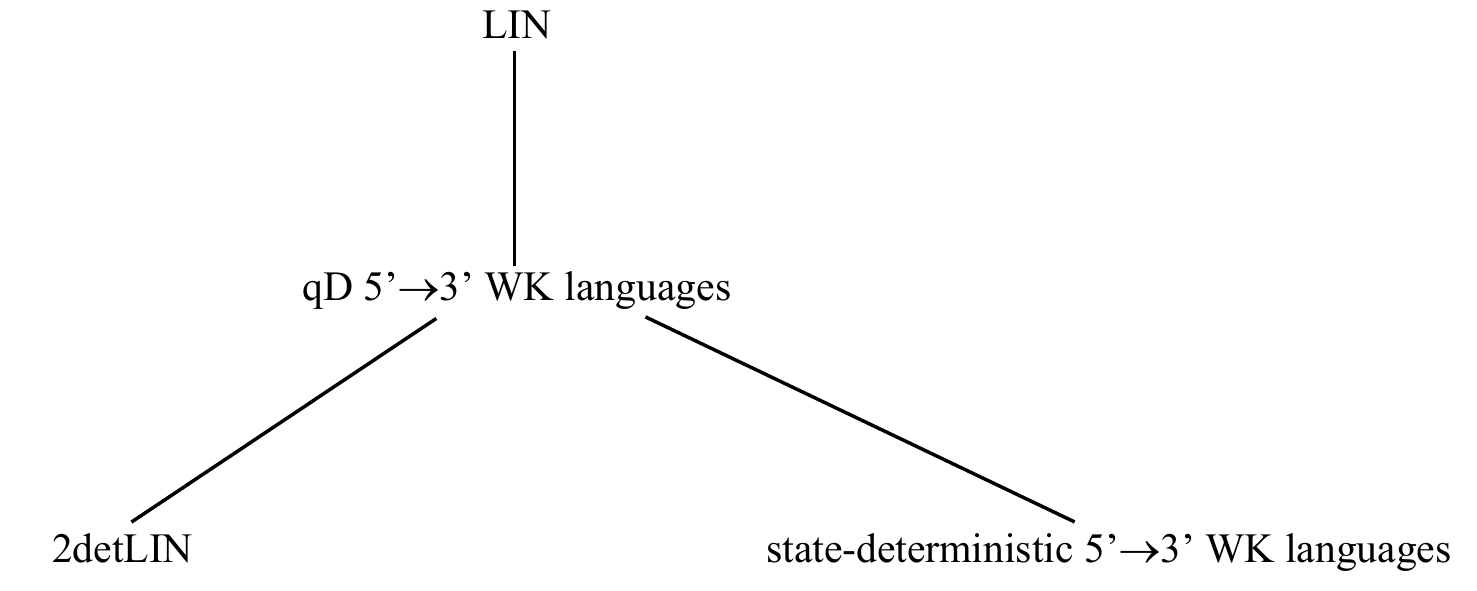} 
            \caption{The place of language class of quasi-deterministic sensing $5' \rightarrow 3'$ WK automata in the hierarchy of related classes of languages.}
    \label{hasse}
\end{figure}

In the rest of the section, we consider some restricted variants of quasi-deterministic
sensing $5'\rightarrow 3'$ WK automata.

\subsection{On stateless variants}

We start this subsection with the following observation.

\begin{proposition}
  Every stateless sensing $5'\rightarrow 3'$ WK automaton is quasi-deterministic.
\end{proposition}
\begin{proof}
In a computation of a stateless automaton, the sole state appears in every configuration, thus the automaton must be quasi-deterministic.
\end{proof}

As the quasi-determinism is a condition that does not have any influence on the other usual restrictions, we can state the following consequences.

\begin{corollary}
 The class of N sensing $5'\rightarrow 3'$ WK automata is the same as the class of qDN sensing $5'\rightarrow 3'$ WK automata.

 The class of NS sensing $5'\rightarrow 3'$ WK automata is the same as the class of qDNS sensing $5'\rightarrow 3'$ WK automata.

 The class of N1 sensing $5'\rightarrow 3'$ WK automata is the same as the class of qDN1 sensing $5'\rightarrow 3'$ WK automata.
\end{corollary}

Based on the results proven in \cite{NaCo-Shag,AFL}, we can establish the following hierarchy results.

\begin{corollary}
 The class of languages accepted by qDN sensing $5'\rightarrow 3'$ WK automata is a proper superset of the class of languages
accepted by qDNS sensing $5'\rightarrow 3'$ WK automata.

Further,  the class of languages accepted by qDNS sensing $5'\rightarrow 3'$ WK automata is a proper superset of the
  class of languages accepted by qDN1 sensing $5'\rightarrow 3'$ WK automata.
\end{corollary}

Further, by observing that the language of Example \ref{ex-bab+b}
cannot be accepted by a quasi-deterministic stateless sensing $5'\rightarrow 3'$ WK automata),
we state that to be stateless is stronger restriction than to be quasi-deterministic.

\begin{proposition}\label{Nsd-1}
The class of languages accepted by qDN sensing $5'\rightarrow 3'$ WK automata
is properly included in the set of languages accepted by quasi-deterministic sensing $5'\rightarrow 3'$ WK automata.
\end{proposition}

\subsection{Relation to regular languages}

In this section, we analyse the relation of the class REG of regular languages to variants of
quasi-deterministic sensing $5'\rightarrow 3'$ WK automata.

On the one hand, we have:
\begin{proposition}
The language class accepted by qDNS sensing $5'\rightarrow 3'$ WK automata is a proper subset of the class REG of regular languages.
\end{proposition}
\begin{proof}
The qDNS sensing $5'\rightarrow 3'$ WK automata can accept only special regular languages
of the form $(v_1+\dots+v_i)^*(u_1+\dots+u_j)^*$, where the transitions in which the first head can read are $\{q\}=\delta(q,v_1,\lambda)=\dots=\delta(q,v_i,\lambda)$ and the transitions with the second head are $\{q\}=
\delta(q,\lambda,u_1)=\dots=\delta(q,\lambda,u_j)$.
Furthermore, we have also shown a regular language in Example \ref{ex-bab+b} that cannot be accepted by any
qD stateless sensing $5'\rightarrow 3'$ WK automata.
\end{proof}

On the other hand, there are non-regular languages that are accepted by qDN sensing $5'\rightarrow 3'$ WK automata:
\begin{example}\label{ex-ab}
  Let us consider the qDN sensing $5'\rightarrow 3'$ WK automaton with only one transition,
  $\{q\} =\delta(q,a,b)$. This automaton accept the non-regular language $\{a^nb^n~|~n\geq 0\}$.
\end{example}

\begin{corollary}
  The class of languages accepted by qDN sensing $5'\rightarrow 3'$ WK automata and the class REG of regular languages are incomparable under set-theoretic inclusion relation.
\end{corollary}

Further, we have that all regular languages are accepted by 1-limited variant:

\begin{proposition}\label{1>REG}
  The class of languages accepted by qD1 sensing $5'\rightarrow 3'$ WK automata properly includes the class REG of regular languages.
\end{proposition}
\begin{proof}
First, %
 the inclusion is obviously coming from the fact that all DFA can be seen as a special variant of qD1 sensing $5'\rightarrow 3'$ WK automata, where the second head does not read any input symbol in any step. %

To show that the inclusion is proper, let us consider the following qDF1 sensing $5'\rightarrow 3'$ WK automaton: $(\{q,p,r\},\{a,b\},q,\{q,p,r\},\delta)$ with transition function $\delta$ as follows:

\noindent $\{p\} = \delta(q,a,\lambda)$ \quad $\{q\} = \delta(p,\lambda,a)$ \quad
$\{r\} = \delta(q,b,\lambda)$ \quad $\{q\} = \delta(r,\lambda,b)$.

\noindent This automaton accepts the language of palindromes over $\{a,b\}$, i.e., exactly those words that are the same if one reads them backward from the end. This is a well-known non-regular language.
\end{proof}

Now, we have already seen some languages, e.g., Example \ref{ex-ab} and in the previous proof, that are not regular but accepted by quasi-deterministic
all-final  sensing $5'\rightarrow 3'$ WK automata (maybe even by a more restricted variant).
Now we show another fact.
\begin{lemma}\label{lm-bab}
The regular language $L'$ given by $b^*ab^*$ is not accepted by any qDF sensing $5'\rightarrow 3'$ WK automata.
\end{lemma}
\begin{proof}
The proof goes by contradiction. Let us assume that $A$ is a qDF sensing $5'\rightarrow 3'$ WK automaton that accepts $L'$. Let $k$ is the length of the longest string that can be read by a transition of $A$. Consider the word $w=b^{2k+1}ab^{2k+1}$ which is in $L'$. Let us consider its accepting computation.
The all-final $A$ reads $k_1$ (at most $k$) $b$-s from the prefix and $k_2$ (at most $k$) $b$-s from the suffix of $w$ in the first step of the computation. However, as all states are final, this leads that the input word $b^{k_1} b^{k_2}$ would also be accepted, however it is not in $L'$.
  By this contradiction the lemma is proven.
  \end{proof}

Based on examples mentioned and the previous lemma we have obtained the following result:

\begin{corollary}
  The class of languages accepted by qDF sensing $5'\rightarrow 3'$ WK automata and the class REG of regular languages are incomparable under set-theoretic inclusion relation.

  The class of languages accepted by qDFS sensing $5'\rightarrow 3'$ WK automata and the class REG of regular languages are incomparable under set-theoretic inclusion relation.

  The class of languages accepted by qDF1 sensing $5'\rightarrow 3'$ WK automata and the class REG of regular languages are incomparable under set-theoretic inclusion relation.
\end{corollary}

\subsection{Further hierarchy results}

In this section, we show further hierarchy results among the analysed classes of languages.

\begin{proposition}\label{pr:N1-F1} %
The class of languages accepted by qDF1 sensing $5'\rightarrow 3'$ WK automata properly includes
the class of languages accepted by qDN1 sensing $5'\rightarrow 3'$ WK automata.
\end{proposition}
\begin{proof}
  The inclusion comes directly from the definition. We need to prove the properness.
  Let us consider the regular language
 $ba^*+a^*$. A qDF1 sensing $5'\rightarrow 3'$ WK automaton that accepts it can be given as
 $(\{q,p\},\{a,b\},q,\{q,p\},\delta)$, with
 $\{q\}=\delta(q,\lambda,a)$ and $\{q\}=\delta(q,\lambda,b)$.

 Now, we argue that no qD stateless sensing $5'\rightarrow 3'$ WK automata can accept this language. In fact, to accept the language by a stateless variant, there must be a transition from its sole state to itself that allows to read a letter $b$ (or a string that contains it).
 However, then by the iterative use of this transition, words containing more than one $b$ would also be accepted.
\end{proof}

\begin{proposition}
The class of languages accepted by qDFS sensing $5'\rightarrow 3'$ WK automata properly includes
the class of languages accepted by qDF1 sensing $5'\rightarrow 3'$ WK automata.
\end{proposition}
\begin{proof}
The inclusion is clear from the definition. To show its properness, let us consider
the following qDFS sensing $5'\rightarrow 3'$ WK automaton:
 $(\{p,q\},\{a,b\},q,\{p,q\},\delta)$ with $\{p\}=\delta(q,aa,\lambda)$, $\{q\}=\delta(p,\lambda,bb)$. It accepts the language $\{a^{2n}b^{2n}~|~n\geq 0\}\cup\{a^{2(n+1)}b^{2n}~|~ n\geq 0\}$.

 On the other hand, it is clear that if a qDF1 sensing $5'\rightarrow 3'$ WK automaton
 accepts a word $w$ of length $k$, then its accepting computation contains exactly $k$ steps, and there are words in the accepted language by each positive integer length up to $k$ based on the given accepting computation in $w$ (composed by the read prefix and read suffix of the input word).
 Since, e.g., in the previous language, there is a word with length $2$, but there is no word with length $1$, clearly it cannot be accepted by any qDF1 sensing $5'\rightarrow 3'$ WK automata.
\end{proof}

Based on the automata, languages and arguments we used in the previous parts, the following result can also be obtained.
\begin{corollary}
The class of languages accepted by qDF1 sensing $5'\rightarrow 3'$ WK automata and
the class of languages accepted by qDN sensing $5'\rightarrow 3'$ WK automata
are incomparable under set-theoretic inclusion relation. \end{corollary}

Based on Proposition \ref{1>REG} and the previous argument about qDF1 sensing $5'\rightarrow 3'$ WK automata,
we also infer the following inclusion:

\begin{corollary}
The class of languages accepted by qD1 sensing $5'\rightarrow 3'$ WK automata properly includes
the class of languages accepted by qDF1 sensing $5'\rightarrow 3'$ WK automata.
\end{corollary}

\begin{proposition}
The class of languages accepted by qDFS sensing $5'\rightarrow 3'$ WK automata properly includes
the class of languages accepted by qDNS sensing $5'\rightarrow 3'$ WK automata.

The class of languages accepted by qDF sensing $5'\rightarrow 3'$ WK automata properly includes
the class of languages accepted by qDN sensing $5'\rightarrow 3'$ WK automata.
\end{proposition}
\begin{proof}
 The inclusions come from the definitions, we need to prove only their properness.
 The language $ba^*+a^*$ used in the proof of Proposition \ref{pr:N1-F1} is accepted by a
 qDF1 sensing $5'\rightarrow 3'$ WK automaton, which is also a qDFS and a qDF sensing $5'\rightarrow 3'$ WK automaton. On the other hand, this language cannot be accepted by any qD stateless sensing $5'\rightarrow 3'$ WK automaton by the argument used in the second part of the proof of
 Proposition \ref{pr:N1-F1}.
\end{proof}

\begin{proposition}
The class of languages accepted by qDF sensing $5'\rightarrow 3'$ WK automata properly includes
the class of languages accepted by qDFS sensing $5'\rightarrow 3'$ WK automata.
\end{proposition}
\begin{proof}
  The inclusion is straightforwardly coming from the definition, we shall prove only its properness.
Consider the separating language $aaa(ab)^*bbb$. %
On the one hand, the
qDF sensing $5'\rightarrow 3'$ WK automaton  $(\{p,q\},\{a,b\},q,\{p,q\},\delta)$ with $\{p\}=\delta(q,aaa,bbb)$, $\{p\}=\delta(p,ab,\lambda)$ %
 accepts it.
On the other hand, in a %
qDFS automaton, in a transition from the initial state only one of the heads can read some input symbols. The shortest nonempty word of the language is $aaabbb$, thus either this is read by one of the heads, or if shorter word is read, then the automaton will accept that prefix or suffix of $aaabbb$ leading to a contradiction.
However, if the whole $aaabbb$ is read by one of the heads, then this must be the prefix (if the left head is used) or the suffix (if the right head is used) of the other words that are accepted. Since this is not true for the language, there could not be any qDFS automaton that accepts it.
\end{proof}

By definition of the restricted variants, we also know the following:
\begin{corollary}
The class of languages accepted by qDS sensing $5'\rightarrow 3'$ WK automata includes
the class of languages accepted by qD1 sensing $5'\rightarrow 3'$ WK automata.
\end{corollary}
We left open whether the inclusion in the above statement is proper.

\begin{proposition}
The class of languages accepted by qDS sensing $5'\rightarrow 3'$ WK automata properly includes
the class of languages accepted by qDFS sensing $5'\rightarrow 3'$ WK automata.
\end{proposition}
\begin{proof}
  The inclusion is a direct consequence of the definition, we shall prove only its properness.
  Let us consider the language $b^*ab^*$. As, it is a regular language, on the one hand, it is easy to give a qD1 sensing $5'\rightarrow 3'$ WK automaton that accepts it (and this automaton is also a qDS sensing $5'\rightarrow 3'$ WK automaton at the same time):
  $(\{p,q\},\{a,b\},q,\{p\},\delta)$ with $\{q\}=\delta(q,b,\lambda)$, $\{p\}=\delta(q,a,\lambda)$ and $\{p\}=\delta(p,b,\lambda)$.

  However, on the other hand, the language $b^*ab^*$ cannot be accepted by any
  qDF sensing $5'\rightarrow 3'$ WK automata (and thus cannot be accepted by any
   qDFS sensing $5'\rightarrow 3'$ WK automata) as we have seen in Lemma \ref{lm-bab}.
\end{proof}

Finally, we state a relation concerning 2detLIN, the class of languages accepted by deterministic sensing $5'\rightarrow 3'$ WK automata.
\begin{theorem}
The class of languages accepted by qD1 sensing $5'\rightarrow 3'$ WK automata  includes
the class 2detLIN of languages.
\end{theorem}
\begin{proof}
  It has been shown in \cite{ActInf2020} that 2detLIN is also accepted by the class of deterministic 1-limited sensing $5'\rightarrow 3'$ WK automata.
On the other hand, as an implication of their definitions, each D1 sensing $5'\rightarrow 3'$ WK automaton is a
qD1 sensing $5'\rightarrow 3'$ WK automaton.
\end{proof}

In fact, D1 sensing $5'\rightarrow 3'$ WK automata are those
qD1 sensing $5'\rightarrow 3'$ WK automata in which exactly one of the heads are allowed to read in each state.
 As one may easily design a
qD1 sensing $5'\rightarrow 3'$ WK automaton that does not satisfy this property, this latter model may be more powerful than the deterministic sensing $5'\rightarrow 3'$ WK automata. We leave this question open.

\section{Conclusions}
\begin{figure}[tb]
    \centering
       \ \ \ \includegraphics[scale=.8]{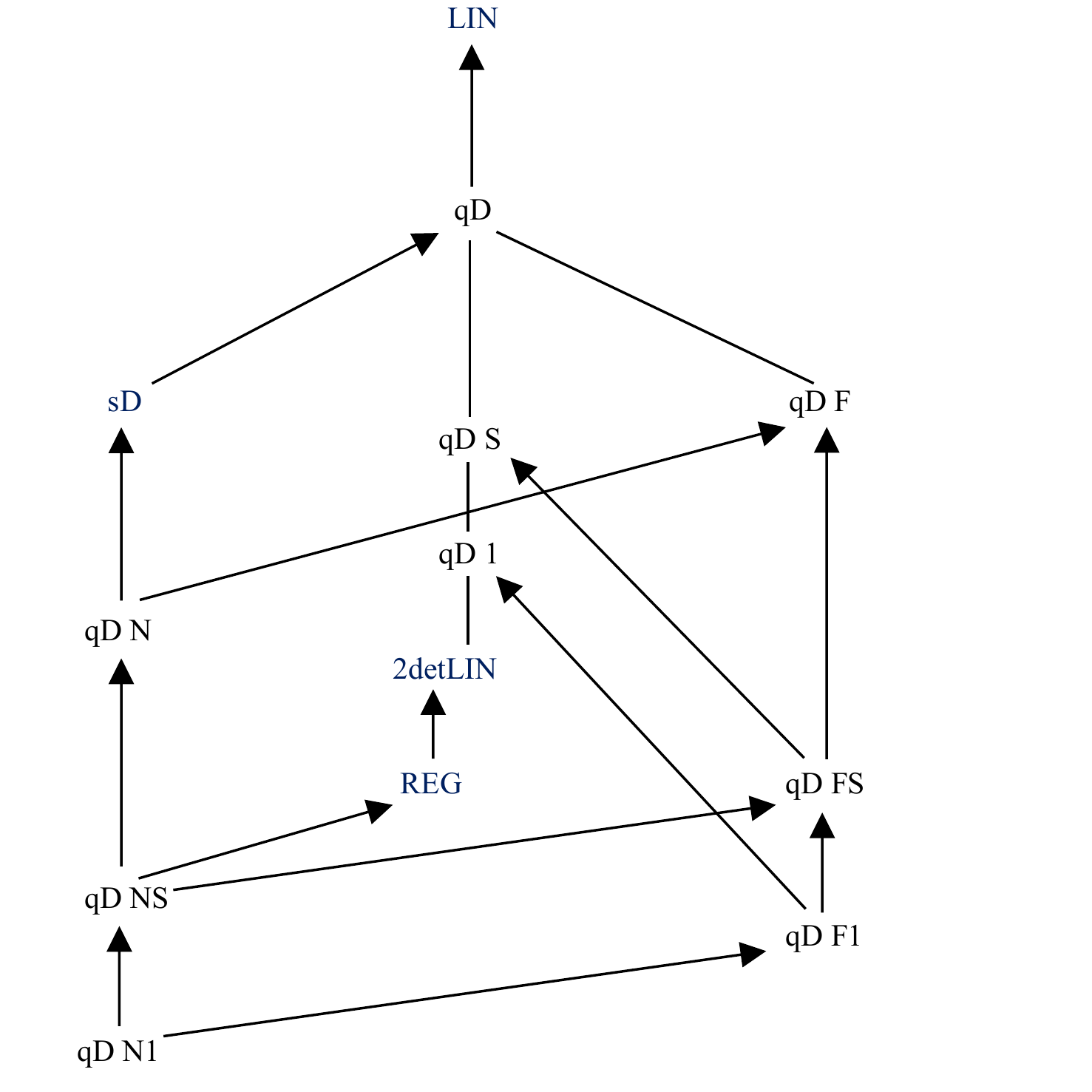} 
            \caption{A hierarchy of the language classes of quasi-deterministic sensing $5' \rightarrow 3'$ WK automata in a Hasse diagram. The abbreviations qD and sD refer to the language classes accepted by state-deterministic and quasi-deterministic $5' \rightarrow 3'$ WK automata; S, F, N, 1 and their combinations are used to
            abbreviate the restricted variants of the qD sensing $5' \rightarrow 3'$ WK automata.
            Arrows show proper inclusions, while lines without arrow head show inclusions where the properness is left open. Blue color shows other related language classes.}
    \label{hasseFULL}
\end{figure}

We have considered a kind of generalisation of determinism in the case of finite state machines. The quasi-determinism may allow such nondeterminism that is based on the input being processed during a computation step. In a quasi-deterministic automaton, the state of the next configuration in the computation is determined, but the next configuration itself may not be.
This new type of determinism also generalise the recently introduced  state-determinism.
 We have shown that quasi-determinism is entirely the same as determinism in the case of $\lambda$-transition free NFA. On the other hand, for WK automata, especially, for sensing $5'\to 3'$ WK automata, because of the two heads and the string-reading feature, it is more interesting.
Knowing that the family of sensing $5'\to 3'$ WK automata accept the family of linear context-free languages, we have shown that a new sublinear  language class is accepted
by the new model, that is a proper superclass of the class 2detLIN of languages accepted by deterministic sensing $5'\to 3'$ WK automata.
We have also studied various restricted classes, and proved various hierarchy results among them, for their summary see Figure \ref{hasseFULL}. The properness of some inclusions, as well as most of the other properties of the defined language classes, including, e.g.,
 closure properties, are left for the future.
 The quasi-determinism may also be expanded to various other types of automata, including
 models based on $5'\to 3'$ WK automata, e.g., automata with multiple runs \cite{Leupold}, jumping $5'\to 3'$ WK automata \cite{NCMA-53jump,AcIn-53jump}, 2-head/linear automata with translucent letters \cite{2h-tr,lin-tr} and
$5'\to 3'$ WK transducers \cite{WK-trans,WK-trans2}.

\section*{Acknowledgements}
Comments of the anonymous reviewers are gratefully acknowledged.

\providecommand{\urlalt}[2]{\href{#1}{#2}}
\providecommand{\doi}[1]{doi:\urlalt{http://dx.doi.org/#1}{#1}}

\end{document}